%% file: journal.tex
\documentclass[11pt]{article}
\usepackage{microtype}
\usepackage[utf8]{inputenc}
\usepackage{lineno}
\usepackage[symbol]{footmisc}
\usepackage{epsfig}
\usepackage{microtype}
\usepackage{graphicx}
\usepackage{enumitem}
\usepackage[linesnumbered,ruled, lined]{algorithm2e}
\usepackage{algpseudocode}
\usepackage{epsfig,enumerate,amsmath,amsfonts,amssymb,amsthm,mathrsfs,ifpdf}
\usepackage{indentfirst,relsize}
\usepackage[numbers]{natbib}
\usepackage{setspace}
\usepackage{enumerate}
\usepackage{latexsym}
\usepackage{stackrel}
\usepackage[all]{xy}
\usepackage[margin = 2.5cm]{geometry}
\usepackage[usenames,dvipsnames]{pstricks}
\usepackage{pst-grad} 
\usepackage{pst-plot} 
\usepackage{xspace}
\usepackage{hyperref}

\newcommand{\size}[1]{\left| #1 \right|}

\newcommand{\E}{\mathbb{E}}
\newcommand{\V}{\mathbb{V}}
\newcommand{\remove}[1]{}

\newcommand{\R}{\mathbb{R}}
\newcommand{\N}{\mathbb{N}}

\newcommand{\cE}{\mathcal{E}}

\newcommand{\cP}{\mathcal{P}}

\newcommand{\cR}{\mathcal{R}}

\newcommand{\Oh}{\mathcal{O}}
\newcommand{\tOh}{\tilde{\mathcal{O}}}
\newcommand{\cF}{\mathcal{F}}

\newcommand{\cH}{\mathcal{H}}

\newcommand{\eps}{\varepsilon}

\newcommand{\complain}[1]{\textcolor{red}{#1}}

\input{macro.tex}

\theoremstyle{plain}
\newtheorem{theo}{Theorem}[section]
\newtheorem{lem}[theo]{Lemma}
\newtheorem{pre}[theo]{Proposition}
\newtheorem{coro}[theo]{Corollary}

\newtheorem{cl}[theo]{Claim}

\theoremstyle{definition}
\newtheorem{defi}[theo]{Definition}
\newtheorem{rem}{Remark}
\newtheorem{obs}[theo]{Observation}

\begin{document}
\title{Faster Counting and Sampling Algorithms using  Colorful Decision Oracle} 

\author{
	Anup Bhattacharya~\footnote{National Institue of Science Education and Research, Bhubaneswar, Inida} 
\and Arijit Bishnu~\footnote{Indian Statistical Institute, Kolkata, India}
\and
Arijit Ghosh~\footnote{Indian Statistical Institute, Kolkata, India}
\and
Gopinath Mishra\footnote{University of Warwick,UK}}

\date{}
\maketitle

\maketitle
\begin{abstract}

\input{abstract.tex}

\end{abstract}

\input{intro2.tex}

\input{prelim.tex}
\input{coarse}


\input{whyimp.tex}


\bibliographystyle{alpha}
\bibliography{reference}
\newpage

\appendix
\input{appendix.tex}


\end{document}

%% file: macro.tex
\newcommand{\dsubset}{A_1,\ldots, A_d}
\newcommand{\ssubset}{A_1^{[a_1]},\ldots, A_s^{[a_s]}}

\newcommand{\gpis}{{\sc CID}\xspace}
\newcommand{\gonepis}{$\mbox{\gpis}_1$ }
\newcommand{\gtwopis}{$\mbox{\gpis}_2$ }
\newcommand{\bis}{{\sc BIS}\xspace}
\newcommand{\tis}{{\sc TIS}\xspace}

\newcommand{\hest}{{\sc $d$-Hyperedge-Estimation}\xspace}
\newcommand{\hsample}{{\sc $d$-Hyperedge-Sample}\xspace}

\newcommand{\test}{{\sc Triangle-Estimation}\xspace}
\newcommand{\pr}{{\mathbb{P}}\xspace}
\newcommand{\verest}{{\sc Verify-Estimate}\xspace}

\newcommand{\lbeps}{\left({n^{-d}\log ^{5d+5} n} \right)^{1/4}}

\renewcommand{\hat}{\widehat}
\newcommand{\etal}{{\it{et al.}}}

\newcommand{\cAc}{\mbox{{\sc Rough Estimation}}\xspace}
\newcommand{\poly}{n^{\Omega(d)}}

\newcommand{\defproblem}[3]{
  \vspace{1mm}
\noindent\fbox{
  \begin{minipage}{0.96\textwidth}
  \begin{tabular*}{\textwidth}{@{\extracolsep{\fill}}lr} #1 \\ \end{tabular*}
  {\bf{Input:}} #2  \\
  {\bf{Output:}} #3
  \end{minipage}
  }
  \vspace{1mm}
}

%% file: abstract.tex
In this work, we consider $d$-{\sc Hyperedge Estimation} and $d$-{\sc Hyperedge Sample} problem in a hypergraph $\cH(U(\cH),\cF(\cH))$ in the query complexity framework, where $U(\cH)$ denotes the set of vertices and $\cF(\cH)$ denotes the set of hyperedges. The oracle access to the hypergraph is  called {\sc Colorful Independence Oracle} ({\sc CID}), which takes $d$ (non-empty) pairwise disjoint subsets of vertices $\dsubset \subseteq U(\cH)$ as input, and answers whether there exists a hyperedge in $\cH$ having (exactly) one vertex in each $A_i, i \in \{1,2,\ldots,d\}$. The problem of $d$-{\sc Hyperedge Estimation} and $d$-{\sc Hyperedge Sample} with {\sc CID} oracle access is important in its own right as a combinatorial problem. Also, Dell {\it{et al.}}~[SODA '20] established that  {\em decision} vs {\em counting} complexities of a number of combinatorial optimization problems can be abstracted out as $d$-{\sc Hyperedge Estimation} problems with a {\sc CID} oracle access.

The main technical contribution of the paper is an algorithm that estimates $m= \size{\cF(\cH)}$ with $\hat{m}$ such that 
{
$$
    \frac{1}{C_{d}\log^{d-1} n} \;\leq\; \frac{\hat{m}}{m} \;\leq\; C_{d} \log ^{d-1} n .
$$ 
by using at most $C_{d}\log ^{d+2} n$ many {\sc CID} queries, where $n$ denotes the number of vertices in the hypergraph $\cH$ and $C_{d}$ is a constant that depends only on $d$}. Our result coupled with the framework of Dell {\it{et al.}}~[SODA '21] implies improved bounds for the following fundamental problems:

\begin{description}
    \item[Edge Estimation] using the {\sc Bipartite Independent Set} ({\sc BIS}). We improve the bound obtained by Beame {\it{et al.}}~[ITCS '18, TALG '20].
    \vspace{2.5pt}
    \item[Triangle Estimation] using the {\sc Tripartite Independent Set} ({\sc TIS}). The previous best bound for the case of graphs with low {\em co-degree} (Co-degree for an edge in the graph is the number of triangles incident to that edge in the graph)  was due to Bhattacharya {\it{et al.}}~[ISAAC '19, TOCS '21], and Dell {\it{et al.}}'s result gives the best bound for the case of general graphs~[SODA '21]. We improve both of these bounds. 
        \vspace{2.5pt}
    \item[Hyperedge Estimation \& Sampling] using {\sc Colorful Independence Oracle} ({\sc CID}). We give an improvement over the bounds obtained by Dell {\it{et al.}}~[SODA '21].
\end{description}

%% file: intro2.tex
\section{Introduction} \label{sec:intro}

Estimating different combinatorial structures like edges, triangles and cliques in an unknown graph that can be accessed only through {\em query oracles} is a fundamental area of research in {\em sublinear algorithms}~\cite{Feige06,GoldreichR08,EdenLRS15,EdenRS18}. Different query oracles provide unique ways of looking at the same graph. 
\remove{A \emph{subset query} with a subset $T \subseteq U$ asks whether $S \cap T$ is empty or not.
}
Beame \etal~\cite{BeameHRRS18} introduced an independent set based subset query oracle, named {\sc Bipartite Independent Set} (\bis) query, to estimate the number of edges in a graph using polylogarithmic queries. The \bis query answers a YES/NO question on the existence of an edge between two disjoint subsets of vertices of a graph $G$. The next natural questions in this line of research were problems of estimation and uniform sampling of hyperedges in hypergraphs~\cite{soda/DellLM20,BhattaISAAC,BhattaTOCS}. In this paper, we will be focusing on these two fundamental questions, and in doing so, we will improve all the previous results~\cite{talg/BeameHRRS20,soda/DellLM20,BhattaISAAC,BhattaTOCS}. 
\subsection{Our query oracle, results and the context}
A hypergraph $\cH$ is a \emph{set system} $(U(\cH),\cF(\cH))$, where $U(\cH)$ denotes a set of $n$ vertices and $\cF(\cH)$, a set of subsets of $U(\cH)$, denotes the set of hyperedges. A hypergraph $\cH$ is said to be {\em $d$-uniform} if every hyperedge in $\cH$ consists of exactly $d$ vertices. The cardinality of the hyperedge set is denoted as $m(\cH)=\size{\cF(\cH)}$. 
We will access the hypergraph using the following oracle\footnote{In~\cite{BishnuGKM018}, the oracle is named as {\sc Generalized Partite Independent Set} oracle. Here, we follow the same suit as Dell \etal~\cite{soda/DellLM20} with respect to the name of the oracle.}~\cite{BishnuGKM018}.  

\begin{defi}
[Colorful Independent Set (\gpis)] Given $d$ pairwise disjoint subsets of vertices $\dsubset \subseteq U(\cH)$ of a hypergraph $\cH$ as input, \gpis{} query answers {\sc Yes} if and only if $m(\dsubset) \neq 0$, where $m(\dsubset)$ denotes the number of hyperedges in $\cH$
having exactly one vertex in each $A_i$, where $i \in \{1,2,\ldots,d\}$. 
\label{def:gpis} 
\end{defi}

Note that the earlier mentioned \bis is a special case of \gpis when $d=2$. 
\remove{
Beame \etal~\cite{BeameHRRS18} showed that $(1 \pm \eps)$-approximation~\footnote{The statement ``$a$ is an $1 \pm \eps$ multiplicative approximation of $b$'' means $\size{b-a} \leq \eps \cdot b$.} to the number of edges in a unknown graph $G$ can be obtained by using $\Oh\left(\frac{1}{\eps^4} \log^{14} n\right)$ \bis queries with probability $1-1/\poly(n)$.
\remove{
The \gpis{} query has a \emph{transversal} nature to it. A \emph{transversal}~\cite{Matgeo} of a hypergraph $\cH=(U(\cH), \cF(\cH))$ is a subset $T \subseteq U(\cH)$ that intersects all sets of $\cF(\cH)$. One can see a \gpis{} query as a \emph{transversal query} as it answers if there exists a transversal for the disjoint subsets $\dsubset$ as in Definition~\ref{def:gpis}. (Gopi: may be we can drop this!)}
}
With this query oracle access, we solve the following two problems. 

\defproblem{\hest}{Vertex set $U(\cH)$ of a hypergraph $\cH$ with $n$ vertices, a \gpis{} oracle access to $\cH$, and $\eps \in (0,1)$.}{A $(1 \pm \eps)$-approximation $\widehat{m}$ to $m(\cH)$ with probability $1-1/\poly$.}

Note that {\sc Edge Estimation} problem is a special case of \hest when $d=2$.

\defproblem{\hsample}{Vertex set $U(\cH)$ of a hypergraph $\cH$ with $n$ vertices, a \gpis{} oracle access to $\cH$, and $\eps\in (0,1)$.}{With probability $1-1/\poly$, report a sample from a distribution of hyperedges in $\cH$ such that the probability that any particular hyperedge is sampled lies in the interval $\left[(1-\eps)\frac{1}{m},(1+\eps)\frac{1}{m}\right]$.}

This area started with the investigation of {\sc Edge Estimation} problem by Dell and Lapinskas~\cite{DellL18,toct/DellL21} and Beame \etal~\cite{BeameHRRS18}, then Bhattacharya \etal~\cite{BhattaISAAC,BhattaTOCS} studied \hest for $d=3$, and more recently Dell \etal~\cite{soda/DellLM20} gave algorithms for \hest and \hsample for general $d$. Beame \etal~\cite{BeameHRRS18} showed that {\sc Edge Estimation} problem can be solved using $\Oh\left(\frac{\log^{14} n}{\eps^4} \right)$ \bis queries. Having estimated the number of edges in a graph using \bis{} queries, a very natural question was to estimate the number of hyperedges in a hypergraph using an appropriate query oracle. This extension is nontrivial as two edges in a graph can intersect in at most one vertex but the intersection pattern between two hyperedges in a hypergraph is more complicated. As a first step towards resolving this question, Bhattacharya \etal~\cite{BhattaISAAC,BhattaTOCS} considered \hest in $3$-uniform hypergraphs using {\sc CID} queries.
They showed that when {\em co-degree} of any pair of vertices in a $3$-uniform hypergraph is bounded above by $\Delta$, then one can solve \hest using $\Oh\left(\frac{\Delta^2 \log^{18}n}{\eps^4}\right)$ {\sc CID} queries. Recall that co-degree of two vertices in a hypergraph is the number of hyperedges that contain both vertices. 
Dell \etal~\cite{soda/DellLM20} generalized the results of Beame \etal~\cite{BeameHRRS18} and  Bhattacharya \etal~\cite{BhattaISAAC,BhattaTOCS}, and obtained a  similar (with an improved dependency in terms of $\eps$) result for the \hest problem for general $d$. Apart from \hest problem, they also considered the problem of \hsample.
The results of Dell \etal~\cite{soda/DellLM20} are formally stated in the following proposition:
 \begin{pre}[{\bf Dell et al.~\cite{soda/DellLM20}}]\label{pro:dell}
 \hest and \hsample can be solved by using $\Oh_d\left(\frac{\log ^{4d+8} n}{\eps^2} \right)$ and $\Oh_d\left(\frac{\log ^{4d+12} n}{\eps^2} \right)$ \gpis queries, respectively.~\footnote{Dell \etal~\cite{soda/DellLM20} studied \hest and \hsample where the probability of success is $1-\delta$ for some given $\delta \in (0,1)$, and have showed that \hest and \hsample can be solved by using $\Oh_d\left(\frac{\log ^{4d+7} n}{\eps^2} \log \frac{1}{\delta} \right)$ and $\Oh_d\left(\frac{\log ^{4d+11} n}{\eps^2}\log \frac{1}{\delta} \right)$ \gpis queries, respectively. In Proposition~\ref{pro:dell}, we have taken $\delta=n^{\Oh(d)}.$  But both the results of  Beame \etal~\cite{BeameHRRS18,talg/BeameHRRS20} and Bhattacharya \etal~\cite{BhattaISAAC, BhattaTOCS} are in the high probability regime.
    
In this paper, we work with success probability to be $1-1/\poly$ for simplicity of presentation and compare our results with all previous results in a high probability regime.}
 \end{pre}
 Currently, the best known bound (prior to this work) for solving \hest problem, for general $d$, is due to Dell \etal~\cite{soda/DellLM20}, but note that for constant $\eps \in (0,1)$, Beame \etal~\cite{BeameHRRS18,talg/BeameHRRS20} still have the best bound for the {\sc Edge Estimation} problem.


Our main result is an improved \emph{coarse estimation} technique, named $\cAc$, and is stated in the following theorem. The significance of the coarse estimation technique will be discussed in Section~\ref{ssec:coarse}.  

\begin{theo}[{\bf Main result}] \label{theo:main} 
    There exists an algorithm $\cAc$ that has \gpis query access to a $d$-uniform hypergraph $\cH(U,\cF)$ and returns $\widehat{m}$ as an estimate for $m=\size{\cF(\cH)}$ such that 
    $$
        \frac{1}{C_{d} \log ^{d-1} n} \leq \frac{\widehat{m}}{m}
        \leq C_{d}\log ^{d-1} n
    $$
    with probability at least $1-1/\poly$ using at most $C_{d}\log^{d+2} n$ \gpis queries, where $C_{d}$ is a constant that depends only on $d$ and $n$ denotes the number of vertices in $\cH$.
\end{theo}
Coarse estimation gives a crude polylogarithmic approximation for $m$, the number of hyperedges in $\cH$. This improvement in the coarse estimation algorithm coupled with \emph{importance sampling} and the algorithmic framework of Dell \etal~\cite{soda/DellLM20} gives an improved algorithm for both \hest and \hsample problems. 
   \begin{theo}[{\bf Improved bounds for estimating and sampling}]
   \label{coro:imp}
\hest and \hsample problems can be solved by using $\Oh_d\left(\frac{\log ^{3d+5} n}{\eps^2} \right)$ and $\Oh_d\left(\frac{\log ^{3d+9} n}{\eps^2} \right)$ \gpis queries, respectively.
 \end{theo}
 The details regarding how Theorem~\ref{theo:main} can be used together with the framework of Dell \etal~\cite{soda/DellLM20}
 to prove Theorem~\ref{coro:imp} will be discussed in Section~\ref{sec:whyimp}.
 
 Using Theorem~\ref{coro:imp}, we directly get the following improved bounds for {\sc Edge Estimation} and \hest {in $3$-uniform hypergraph} by substituting $d=2$ and $d=3$, respectively.
 \begin{coro}
    \label{cor-new-final}
    \begin{description}
     \item[(a)]{\sc Edge Estimation} can be solved using $\Oh\left(\frac{\log ^{11} n}{\eps^{2}} \right)$ queries to {\sc Bipartite Independent Set} ({\sc BIS}) oracle.

     \item[(b)] \hest in a $3$-uniform hypergraph can be solved using $\Oh\left(\frac{\log ^{14} n}{\eps^{2}} \right)$ {\sc CID} queries.
    \end{description}
 \end{coro}
  The above corollary gives the best bound (till now) for the {\sc Edge Estimation}. Recall that Bhattacharya \etal~\cite{BhattaISAAC, BhattaTOCS} proved that when the co-degree of a $3$-uniform graph is bounded by $\Delta$ then 
 \hest in that hypergraph can be solved using 
 $\Oh\left(\frac{\Delta^2 \log ^{18} n}{\eps ^4}\right)$ {\sc CID} queries. For fixed $\eps \in (0,1)$ and $\Delta = o(\log n)$ the bound obtained by 
Bhattacharya \etal~\cite{BhattaISAAC, BhattaTOCS} is asymptotically better than the bound we get from Dell et al.~\cite{soda/DellLM20}, see Proposition~\ref{pro:dell}. Note that Corollary~\ref{cor-new-final}~(b) improves the bounds obtained by Bhattacharya \etal~\cite{BhattaISAAC, BhattaTOCS} and Dell et al.~\cite{soda/DellLM20} for all values of $\Delta$ and $\eps \in (0,1)$.

\subsection{Fundamental role of coarse estimation \remove{in Beame \etal~\cite{BeameHRRS18, talg/BeameHRRS20} and Bhattacharya \etal~\cite{BhattaISAAC, BhattaTOCS}~vis-a-vis Dell \etal~\cite{soda/DellLM20}}}
\label{ssec:coarse}
The framework of Dell \etal~\cite{soda/DellLM20} is inspired by the following observation. Let us consider  $t=\Oh\left(\frac{\log n}{\eps^2}\right)$ independent subhypergraphs each induced by $n/2$ uniform random vertices. The probability, that a particular hyperedge is present in a subhypergraph induced by $n/2$ many uniform random vertices, is $\frac{1}{2^d}$. Denoting $X$ as the sum of the numbers of the hyperedges present in the $t$ subhypergraphs, observe that $\frac{2^d}{t}X$ is a $\left(1 \pm \eps\right)$-approximation of $m$. If we repeat the procedure recursively $\Oh(\log n)$ times, then all the subhypergraphs will have a bounded number of vertices in terms of $d$, at which point the number of hyperedges can be determined exactly by using $\Oh_d(1)$ \gpis queries. However, the number of induced subhypergraphs in the worst case can become as {\em large} as $\Omega\left((\log n)^{\log n}\right)$.

To have the number of subhypergraphs bounded at all point of time, they use 
 \emph{importance sampling}. It is about maintaining the weighted sum of some variables whose approximate
value is known to us. The output will be a bounded number of variables and some weight parameters
such that the weighted sum of the variables estimates the required sum. The objective of the
importance sampling procedure in Beame \etal~\cite{BeameHRRS18, talg/BeameHRRS20} and Bhattacharya \etal~\cite{BhattaISAAC, BhattaTOCS}, are also the same~\footnote{In fact, Bhattacharya \etal~\cite{BhattaISAAC, BhattaTOCS} directly use the importance sampling developed by Beame et al.~\cite{BeameHRRS18, talg/BeameHRRS20}}. However, Dell \etal~ 
improved the importance sampling result by the use of a particular form of Bernstein inequality and by a very careful analysis.

To apply importance sampling, it is required to have a rough estimate (possibly with a polylogarithmic approximation factor)  of the number of hyperedges in each subhypergraph that are currently present for processing -- this is what exactly coarse estimation does. The objective of coarse estimation in Beame \etal~\cite{BeameHRRS18, talg/BeameHRRS20} and Bhattacharya \etal~\cite{BhattaISAAC, BhattaTOCS} are also the same~\footnote{Note that the main merit of the framework of Dell \etal~\cite{soda/DellLM20} over Beame \etal~\cite{BeameHRRS18,talg/BeameHRRS20} and Bhattacharya \etal~\cite{BhattaISAAC, BhattaTOCS} is not only that it generalized to hypergraph, but also the dependence on $\eps$ is $1/\eps^2$ in Dell \etal~\cite{soda/DellLM20}'s work as opposed to $\frac{1}{\eps^4}$ in  Beame \etal~\cite{BeameHRRS18,talg/BeameHRRS20} and  Bhattacharya \etal~\cite{BhattaISAAC, BhattaTOCS}.}. 
But all these frameworks have a commonality. The approximation guarantee and the query complexity of the coarse estimation has a direct bearing on the query complexity of the final algorithm. 

\remove{
From the above discussion, it is clear that to improve the results with in the limits of the above discussed frameworks it is important to look at how the approximation factor and query complexity of coarse estimation can be improved, on which this paper focuses on. 
}
\remove{An improved coarse estimation was possibly the primary reason of Dell \etal's improvement over the result of Beame \etal. We also focus on an improvement in coarse estimation.}

Therefore, any improvement in the coarse estimation algorithm will directly improve the query complexities of \hest and \hsample. In this paper, we focus on improving the coarse estimation algorithm.

\subsection{Setup and notations}
\noindent We denote the sets $\{1,\ldots,n\}$ and $\{0,\ldots,n\}$ by $[n]$ and $[n^*]$, respectively. A hypergraph $\cH$ is a \emph{set system} $(U(\cH),\cF(\cH))$, where $U(\cH)$ denotes the set of vertices and $\cF(\cH)$ denotes the set of \remove{unordered} hyperedges. The set of vertices present in a hyperedge $F \in \cF(\cH)$ is denoted by $U(F)$ or simply $F$. A hypergraph $\cH$ is said to be $d$-uniform if all the hyperedges in $\cH$ consist of exactly $d$ vertices. The cardinality of the hyperedge set is $m(\cH)=\size{\cF(\cH)}$. For $\dsubset \subseteq U(\cH)$ (not necessarily pairwise disjoint), $\cF(\dsubset) \subseteq \cF(\cH)$ denotes the set of hyperedges having a vertex in each $A_i$, and $m(\dsubset)$ is the number of hyperedges in $\size{\cF(\dsubset)}$. \remove{For $u \in U(\cH)$, $\cF(u)$ denote the set of \remove{unordered} hyperedges that are incident on $u$. For $u \in U(\cH)$, the degree of $u$ in $\cH$, denoted as $\deg_{\cH}(u)=\size{\cF(u)}$ is the number of hyperedges incident on $u$. For a set $A$ and $a \in \N$, $A,\ldots,A$ ($a$ times) will be denoted as $A^{[a]}$. Let $\dsubset \subseteq U(\cH)$ be such that for every $i,j \in [d]$ either $A_i=A_j$ or $A_i \cap A_j =\emptyset$. 
This has a bearing on the \gpis{} oracle queries we make; either the sets we query with are disjoint, or are the same. Consider the following $d$-partite sub-hypergraph of $\cH$: $\left( U(\dsubset), \cF(\dsubset) \right)$ where the vertex set is $U(\dsubset) =\bigcup_{i=1}^d A_i$ and the hyperedge set is $\cF(\dsubset) = \left\{ \{ i_1, \ldots, i_d \}~|~ i_j \in A_j \right\}$; we will denote this $d$-partite sub-hypergraph of $\cH$ as $\cH(\dsubset)$. With this notation, $\cH\left( U^{[d]} \right)$ makes sense as a $d$-partite sub-hypergraph on a vertex set $U$.
The number of hyperedges in $\cH(\dsubset)$ is denoted by $m(\dsubset)$.}

Let $\E[X]$ and $\V[X]$ denote the expectation and variance of the random variable $X$. For an event $\cE$, the complement of $\cE$ is denoted by $\overline{\cE}$. The statement ``$a$ is a $(1 \pm \eps)$-approximation of $b$'' means $\size{b-a} \leq \eps \cdot b$. For $x \in \R$, $\exp(x)$ denotes the standard exponential function $e^x$. In this paper, $d$ is a constant, and $\Oh_d(\cdot)$ and $\Omega_d(\cdot)$ denote the standard $\Oh(\cdot)$ and $\Omega(\cdot)$, where the constant depends only on $d$. We use $\log ^{k} n$ to denote $(\log n)^k$. By polylogarithmic, we mean {$\Oh_d\left( \frac{\log^{\Oh(d)} n}{{\eps}^{\Omega(1)}}\right)$ in this paper.}

\subsection{Paper organization}
\label{ssec:paperorg}
\noindent
In Section~\ref{sec:prelim_actual}, we describe the notion of an {\em ordered hyperedge}, and define three other query oracles that can be simulated by using $\Oh_d(\log n)$ \gpis queries. The role of ordered hyperedges and these oracles are mostly expository purposes, i.e., they help us to describe our algorithms {and the calculations} more neatly. Section~\ref{sec:over} gives a brief overview of the proof of our main technical result. In Section~\ref{sec:coarse} we give the proof of our main result (Theorem~\ref{theo:main}).  We describe {in Section~\ref{sec:whyimp}} implications of our main result and how Theorem~\ref{theo:main} can be used to prove Theorem~\ref{coro:imp}. The equivalence proofs of the \gpis{} oracle and its variants are discussed in  Section~\ref{sec:prelim_actual}. Some useful probability results are given in Appendix~\ref{sec:prelim}. Since we use different types of oracles in the calculations, we have recalled all their definitions in Appendix~\ref{sec:oracle-def} for the ease of reference. 

%% file: prelim.tex
\section{Preliminaries: Ordered hyperedges, \gpis oracle, and its variants} \label{sec:prelim_actual}

 \paragraph*{Ordered hyperedges} 
 
We will use the subscript $``o"$ to denote the set of ordered hyperedges. For example, $\cH_o(U,\cF_o)$ denotes the ordered hypergraph 
corresponding to $\cH(U,\cF)$. Here $\cF_o(\cH)$ denotes the set of ordered 
hyperedges that contains $d!$ ordered $d$-tuples for each hyperedge in 
$\cH(U,\cF)$. Let $m_o(\cH_o)$ denotes $\size{\cF_o(\cH_o)}$. Note that $m_o(\cH_o)=d!m(\cH)$. Also, let 
$\cF_o(\dsubset)$ denotes the set $\{F_o \in \cF_o(\cH):\mbox{the $i$-th vertex of $F_o$ is in 
$A_i, \forall i \in [d]$}\}$. The corresponding number for ordered hyperedges is 
$m_o(\dsubset)$. Note that $\cF_o(U(\cH),\ldots,U(\cH))=\cF_o(\cH)$. 

\paragraph*{\gpis oracle and its variants}

Note that the \gpis query takes as input $d$ pairwise disjoint subsets of vertices. We now define two related query oracles \gpis{}$ _1$ and \gpis{}$_2$ that remove the disjointness requirements for the input. Then we extent \gtwopis to the ordered setting. We show that both query oracles can be simulated, with high probability, by making $\Oh_d(\log n)$ queries to the \gpis oracle. The oracles \gpis{}$_1$ and \gpis{}$_2$ will be used in the description of the algorithm for ease of exposition.

\begin{description}
	\item[$\mbox{{\sc CID}}_1$:] Given $s$ pairwise disjoint subsets of vertices $A_1,\ldots,A_s \subseteq U(\cH)$ of a hypergraph $\cH$ and $a_1,\ldots,a_s \in [d]$ such that $\sum_{i=1}^{s}a_i =d$, \gonepis query on input $A_1^{[a_1]},A_2^{[a_2]},\cdots,A_s^{[a_s]}$ answers {\sc Yes} if and only if $m(\ssubset) \neq 0$. Here $A^{[a]}$ denotes the set $A$ repeated $a$ times.

	\item[$\mbox{{\sc CID}}_2$:] Given any $d$ subsets of vertices $\dsubset \subseteq U(\cH)$ of a hypergraph $\cH$, \gtwopis query on input $A_1,\ldots,A_d$ answers  {\sc Yes}  if and only if $m(\dsubset) \neq 0$.
	\item[$\mbox{{\sc CID}}_2^o$:] Given any $d$ subsets of vertices $\dsubset \subseteq U(\cH_o)$ of an ordered hypergraph $\cH_o$, $\mbox{{\sc CID}}_2^o$ query on input $A_1,\ldots,A_d$ answers  {\sc Yes}  if and only if $m_o(\dsubset) \neq 0$.
\end{description}

Observe that the \gtwopis query is the same as the \gpis query without the requirement that the input sets are disjoint. For the \gonepis query, multiple repetitions of the same set is allowed in the input. It is obvious that a \gpis query can be simulated by a \gonepis or \gtwopis query. Also, $\mbox{{\sc CID}}_2^o$ is the ordered analogue of \gtwopis. Using the following observation, we show how a $\mbox{{\sc CID}}_2^o$, \gonepis, or a \gtwopis query can be simulated by a polylogarithmic  number of \gpis queries. 

\begin{obs}[{\bf Connection between query oracles}] \label{obs:queryoracles}
Let $\cH(U,\cF)$ denote a hypergraph and $\cH_o(U,\cF_o)$ denote the corresponding ordered hypergraph.
\begin{itemize}
	\item[(i)] A \gonepis query to $\cH(U,\cF)$ can be simulated using $\Oh_d(\log n)$ \gpis queries with probability $1-1/\poly$.
	\item[(ii)] A \gtwopis query $\cH(U,\cF)$ can be simulated using $\Oh_d(1)$ \gonepis queries.
	\item[(iii)] A \gtwopis query $\cH(U,\cF)$ can be simulated using $\Oh_d(\log n)$ \gpis queries with  probability $1-1/\poly$.
	\item[(iv)] A $\mbox{{\sc CID}}_2^o$ query to  $\cH_o(U,\cF_o)$ can be simulated using a \gtwopis query to $\cH(U,\cF)$.
\end{itemize}
\end{obs}{
\begin{proof}
\begin{itemize}
	\item[(i)] Let the input of \gonepis query be $\ssubset$ such that $a_i \in [d]~\forall i \in [s]$ and $\sum\limits_{i=1}^s a_i =d$. We partition each $A_i$ randomly into $a_i$ parts $B_i^j$ for $j \in [a_i]$. We make a \gpis query with input $B_1^{1},\ldots,B_1^{a_1},\ldots, B_s^{1},\ldots,B_s^{a_s}$. Note that 
$$
\cF(B_1^{1},\ldots,B_1^{a_1},\ldots, B_s^{1},\ldots,B_s^{a_s}) \subseteq \cF(\ssubset).
$$

So, if \gonepis outputs {\sc `No'} to query $\ssubset$, then the above \gpis query will also report {\sc `No'} as its answer. If \gonepis answers {\sc `Yes'}, then consider a particular hyperedge $F \in \cF(\ssubset)$. Observe that 
\begin{eqnarray*}
&& \pr(\mbox{\gpis oracle answers {\sc `Yes'}})\\
&\geq& \pr(\mbox{$F$ is present in $\cF(B_1^{1},\ldots,B_1^{a_1}, \ldots \ldots, B_s^{1},\ldots,B_s^{a_s})
$})\\
	&\geq& \prod\limits_{i=1}^s \frac{1}{a_i ^{a_i}} \\ 
	&\geq& \prod\limits_{i=1}^s \frac{1}{d ^{a_i}}~~~~~~~~~~(\because a_i \leq d~\mbox{for all}~i\in [d])\\
	 &=& \frac{1}{d^{d}}~~~~~~~~~~(\because \sum\limits_{i=1}^s a_i =d)
\end{eqnarray*}

We can boost up the success probability arbitrarily by repeating the above procedure polylogarithmic times.
  
	\item[(ii)] Let the input to \gtwopis query be $\dsubset$. Let us partition each set $A_i$ into at most $2^{d-1}-1$ subsets depending on $A_i$'s intersection with $A_j$'s for $j \neq i$. Let $\cP_i$ denote the corresponding partition of $A_i$, $i \in [d]$. Observe that for any $i \neq j$, if we take any $B_i \in \cP_i$ and $B_j \in \cP_j$, then either $B_i=B_j$ or $B_i \cap B_j = \emptyset$.
  
For each $(B_1,\ldots,B_d) \in \cP_1 \times \ldots \times \cP_d$, we make a \gonepis query with input $(B_1,\ldots,B_d)$. Total number of such \gonepis queries is at most $2^{\Oh(d^2)}$, and we report {\sc `Yes'} to the \gtwopis query if and only if at least one \gonepis query, out of the $2^{\Oh(d^2)}$ queries, reports {\sc `Yes'}.
 
 \item[(iii)] It follows from (i) and (ii).
\end{itemize}

\end{proof}
}

%% file: coarse.tex
\section{Overview of the main structural result}
\label{sec:over}


To prove Theorem~\ref{theo:main}, we first consider Lemma~\ref{lem:prob1}, which is the central result of the paper and is the ordered hypergraph analogue of Theorem~\ref{theo:main}.  The main theorem (Theorem~\ref{theo:main}) follows from Lemma~\ref{lem:prob1} along with Observation~\ref{obs:queryoracles}.
\begin{lem}[{\bf Main Lemma}]
\label{lem:prob1}
 There exists an algorithm $\cAc$ that has $\mbox{{\sc CID}}_2^o$ query  access to a $d$-uniform ordered hypergraph $\cH_o(U,\cF_o)$ corresponding to hypergraph $\cH(U,\cF)$ and returns $\widehat{m}_o$ as an estimate for $m_o=\size{\cF_o(\cH_o)}$ such that 
 $$
        \frac{1}{C_{d} \log ^{d-1} n} \leq \frac{\widehat{m}}{m}
        \leq C_{d}\log ^{d-1} n
    $$
    
       with probability at least $1-1/\poly$ using at most $C_{d}\log^{d+1} n$ $\mbox{{\sc CID}}_2^o$ queries, where $C_{d}$ is a constant that depends only on $d$.
\end{lem}

At a high level, the idea for an improved coarse estimation involves a recursive bucketing technique and careful analysis of the intersection pattern of hypergraphs.

\remove{
Let us consider partitioning the vertices in $U_1=U(\cH)$ such that the vertices in each bucket have approximately (upto 2 factor of each other) the same number of hyperedges in $\cH_o$, having the vertices as the first vertex of the ordered hyperedge. So, there will be at most $d \log n +1$ buckets. Observe that there is a bucket $Z_1$ such that  the number of hyperedges, having the vertices in the bucket as the first vertex, is $\frac{m_o}{d \log n +1 }$. Let the number of hyperedges in $\cH_o$ having a vertex of the bucket $Z_1$ hyperedges) lies betwenn $2^{q_1}$ and $2^{q_1+1}-1$. Then we can argue that $\size{Z_1} \geq \frac{m_o}{2^{q_1+1}(d \log n +1)}$. Now consider a vertex $z \in Z_1$, and all the hyperedges in $\cH_o$ having $z$ as the first vertex, and let us partition the vertices in $U_2=U(\cH)$ into buckets such that the vertices in the same bucket have approximately the same number of hyperedges with $z$ as, as the first vertex

 Also, as the vertices in each bucket have approximately (up to 2 factor of each other) same number of hyperedges in $\cH_o$, as the first vertex, the bucket having at 

Note that $m_o=\size{\cF_o(\cH_o)}=\cF_o(U_1,\ldots,U_d)$, where $U_1=\ldots=U_d=U(\cH)$. }
To build up towards the final proof, we need to prove Lemma~\ref{lem:prob1}. Towards this end, we first define some quantities and prove Claim~\ref{clm:verify}. For that, let us think of partitioning the vertex set in $U_1=U(\cH)$ into buckets such that the vertices in each bucket appear as the first vertex in approximately the same number of hyperedges.  So, there will be at most $d \log n +1$ buckets. It can be shown that that there is a bucket $Z_1 \subseteq U_1$  such that  the number of hyperedges, having the vertices in the bucket as the first vertex, is at least  $\frac{m_o}{d \log n +1 }$. {For each vertex $z_1 \in Z_1$, let the number of hyperedges in $\cH_o$, having $z_1$ as the first vertex, lie between $2^{q_1}$ and $2^{q_1+1}-1$ for some suitable $q_1$.} Then we can argue that 
$$
    \size{Z_1} \geq \frac{m_o}{2^{q_1+1}(d \log n +1)}.
$$
Similarly, we extend the bucketing idea to tuples as follows. Consider a vertex $a_1$ in a particular bucket of $U_1$ and consider all the ordered hyperedges in $\cF_o(a_1)$ containing $a_1$ as the first vertex. We can bucket the vertices in $U_2=U(\cH)$  such that the vertices in each bucket of $U_2$ are present in approximately the same number of hyperedges in $\cF_o(a_1)$ as the second vertex. We generalize the above bucketing strategy with the vertices in $U_i$'s, which is formally described below. \remove{Then we state and prove Claim~\ref{clm:verify}.} Notice that this way of bucketing will allow us to use conditionals {on sampling vertices from the desired buckets of $U_i$'s}.

For $q_1 \in [(d \log n)^*]$, let $U_1(q_1) \subseteq U_1$ be the set of vertices in $a_1\in U_1$ such that for each $a_1 \in U_1(q_1)$, the number of hyperedges in $\cF_o(\cH_o)$, containing $a_1$ as the first vertex, lies between $2^{q_1}$ and  $ 2^{q_1+1}-1$. For $2 \leq i \leq d-1$, and $q_j \in [(d \log n)^*]$ for each $j \in [i-1]$, consider $a_1 \in U_1(q_1), a_2 \in U_2((q_1,a_1),q_2), \ldots,$ $ a_{i-1} \in U_{i-1}((q_1,a_1),\ldots,(q_{i-2},a_{i-2}),q_{i-1}) $. Let $U_i((q_1,a_1),\ldots, (q_{i-1},a_{i-1}),q_i)$ be the set of vertices in $U_i$ such that for each $u_i \in U_i((q_1,u_1),\ldots, (q_{i-1},a_{i-1}),q_i) $, the number of ordered hyperedges in $\cF_o(\cH_o)$, containing $u_j$ as the $j$-th vertex for all $j \in [i]$, lies between $2^{q_i}$ and  $ 2^{q_i+1}-1$. We need the following result to proceed further. For ease of presentation, we use $(Q_i,A_i)$ to denote $(q_1,a_1),\ldots, (q_{i-1},a_{i-1})$ for $ 2 \leq i \leq d-1$. {Informally, Claim~\ref{clm:verify} says that for each $i \in [d-1]$, there exists a bucket in $U_i$ having a \emph{large} number of vertices contributing approximately the same number of hyperedges.}.

\begin{cl}\label{clm:verify}
\begin{description}
	\item[(i)] There exists $q_1 \in [(d \log n)^*]$ such that 
	$$
	    \size{U_1(q_1)} > 
	    \frac{m_o(\cH_o)}{2^{q_1+1}(d\log n +1)}.
	$$
	\item[(ii)] Let $2 \leq i \leq d-1$ and $q_j \in [(d \log n)^*]~\forall j \in [i-1]$. Let $a_1 \in U_1(q_1)$, $a_{j} \in U_{j}((Q_{j-1},A_{j-1}),q_{j})$ $\forall j \neq 1$ and $j<i$. There exists $q_i \in [(d \log n)^*]$ such that $$\size{ U_i((Q_i,A_i),q_i)} > \frac{2^{q_{i-1}}}{2^{q_i +1}(d\log n +1)}.$$
\end{description}
\end{cl}
{\begin{proof}
\begin{itemize}
	\item[(i)] Observe that $m_o(\cH_o)=\sum\limits_{q_1=0}^{d\log n } m_o(U_1(q_1),U_2,\ldots,U_d)$. So, 
there exists $q_1 \in [(d \log n)^*]$ such that $m_o(U_1(q_1),U_2,\ldots,U_d) \geq \frac{m_o(\cH_o)}{d\log n + 1}$. From the definition of $U_1(q_1)$, $ m_o(U_1(q_1),U_2,\ldots,U_d) < \size{U_1(q_1)} \cdot 2^{q_1 +1}$. Hence, there exists $q_1 \in [(d \log n)^*]$ such that 
$$ 
	\size{U_1(q_1)} > \frac{m_o(U_1(q_1),U_2,\ldots,U_d)}{2^{q_1+1}} 
	\geq  \frac{m_o(\cH_o)}{2^{q_1+1} (d\log n+1)}.
$$
	\item[(ii)] {Note that}
	\begin{eqnarray*}
	 && m_o(\{a_1\},\ldots, \{a_{i-1}\},U_i,\ldots,U_d)\\
	&=& \sum_{q_i=0}^{d\log n } m_o(\{a_1\},\ldots, \{a_{i-1}\},U_i((Q_{i-1},A_{i-1}),q_i),\ldots,U_d).
	\end{eqnarray*} So, there exists $q_i \in [(d \log n)^*]$ such that 
	\begin{eqnarray*}
&& m_o(\{a_1\},\ldots, \{a_{i-1}\},U_i((Q_{i-1},A_{i-1}),q_i),\ldots,U_d) \\
&&~~~~~~~~~~~~~~~~~~~~ \geq \frac{m_o(\{a_1\},\ldots, \{a_{i-1}\},U_i,\ldots,U_d)}{d\log n +1}.
\end{eqnarray*}
From the definition of $U_i((Q_{i-1},A_{i-1}),q_i)$, we have $$m_o(\{a_1\},\ldots, \{a_{i-1}\},U_i((Q_{i-1},A_{i-1}),q_i),\ldots,U_d)  < \size{U_i((Q_{i-1},A_i),q_i)} \cdot 2^{q_i +1}$$ Hence, there exists $q_i \in [(d \log n)^*]$ such that 
\begin{eqnarray*}
	\size{U_i((Q_{i-1},A_i),q_i)} &>& \frac{m_o(\{a_1\},\ldots, \{a_{i-1}\},U_i((Q_{i-1},A_{i-1}),q_i),\ldots,U_d\})}{2^{q_i +1}}\\
	&\geq& \frac{m_o(\{a_1\},\ldots, \{a_{i-1}\},U_i,\ldots,U_d\})}{2^{q_i +1}(d\log n +1)}\\
	&\geq& \frac{2^{q_{i-1}}}{2^{q_i +1}(d\log n +1)}
\end{eqnarray*}
\end{itemize}
\end{proof}
 
 From Claim~\ref{clm:verify}, it follows that there exists $(q_1,\ldots,q_{d-1}) \in [(d 
 \log n)^{*}]^{d-1}$ such that  $\size{U_1(q_1)} > \frac{m_o(\cH_o)}{2^{q_1+1}
 (d\log n +1)}$ and $\size{ U_i((Q_i,A_i),q_i)} > \frac{2^{q_{i-1}}}{2^{q_i +1}
 (d\log n +1)}$. So, if we sample each vertex in $U_1$ with probability 
 $p_1=\min\{\frac{2^{q_1}}{m_o},1\}$ independently to generate $B_1$, each 
 vertex of $U_{i}$ ($2 \leq i \leq d-1$) with probability $p_i=\min\{2^{q_{i} - j_{i-1}} \cdot d\log n,1\}$ independently to generate $B_i$, and each vertex in $U_d$ with probability $\min\{2^{-q_{d-1}},1\}$ to generate $B_d$, then we can show that $\cF_o(B_1,\ldots,B_d)$ is nonempty with probability at least  $\prod_{i=1}^d p_i \geq \frac{1}{2^d}$. The success probability $\frac{1}{2^d}$ can be amplified by repeating the procedure suitable number of times. So, if we consider all possible $\Oh_d\left(\log ^{d-1} n\right)$ guesses for $(q_1,\ldots,q_{d-1})$, we have that there exists a guess for which  $\cF_o(B_1,\ldots,B_d)$  is nonempty, that is $m_o(B_1,\ldots,B_d)\neq 0$, and that can be determined by a $\mbox{{\sc CID}}_2^o$ query. In total, there will be $\Oh_d(\log ^{d-1} n)$ $\mbox{{\sc CID}}_2^o$ queries.
 
 However, the sampling probability $p_1$ to sample the vertices from $U_1$ depends on $m_o$. But we do not know $m_o$. Observe that the above procedure works even if we know any lower bound on $m_o$. So, the idea is to consider geometrically decreasing guesses for $m_o$ starting from $m_o=n^d$, and call the above procedure for the guesses until the $\mbox{{\sc CID}}_2^o$ query corresponding to the guess $\hat{\cR}$ for $m_o$ reports that $m_o(B_1,\ldots,B_d)\neq 0$. We will be able to achieve the desired result by showing that, for any guess at least a polylogarithmic factor more than the correct $m_o$, the corresponding $\mbox{{\sc CID}}_2^o$ queries over the samples report $m_o(B_1,\ldots,B_d)\neq 0$ with probability at most $\Oh\left(\frac{1}{2^d}\right)$. The success probability of $1-\Omega\left(\frac{1}{2^d}\right)$ can be amplified by repeating the procedure suitable number of times for each guess. In the next section, we formalize the discussion in this section.
 }
\section{Proof of Lemma~\ref{lem:prob1}} \label{sec:coarse}

 We now prove Lemma~\ref{lem:prob1} formally. The algorithm corresponding to Lemma~\ref{lem:prob1} is Algorithm~\ref{algo:coarse} (named $\cAc$). Algorithm~\ref{algo:verify} (named \verest) is a subroutine of Algorithm~\ref{algo:coarse}. Algorithm~\ref{algo:verify} determines whether a given estimate $\hat{R}$ {of the number of ordered hyperedges} is correct up to $\Oh_d(\log ^{2d-3} n)$ factor. Lemma~\ref{lem:coarse1} and~\ref{lem:coarse2} are intermediate results needed to prove Lemma~\ref{lem:prob1}; they bound the probability from above and below, respectively of \verest accepting the estimate $\hat{\cR}$.

\begin{lem} \label{lem:coarse1} If $\hat{\cR} \geq 20d^{2d-3}4^d ~ m_o(\cH_o) \log^{2d-3} n$, then 
	$$ \pr(\mbox{\verest($\cH_o,\hat{\cR}$) accepts {the estimate $\hat{R}$}}) \leq \frac{1}{20 \cdot 2^d}.$$
\end{lem}

\begin{algorithm}[H]

\SetAlgoLined

\caption{\verest($\cH_o,\hat{\cR}$)}
\label{algo:verify}
\KwIn{\gpis query access to a $d$-uniform hypergraph $\cH_o(U, \cF)$ and a guess $\widehat{R}$ for the number of hyperedges in $\cH_o$.}
\KwOut{{\sc Accept} $\hat{\cR}$ or {\sc Reject} $\hat{\cR}$.}
Let

$~~~~~U_1=\ldots=U_d=U(\cH)$
\For{($j_1= d\log n$ to $0$)}
{
find $B_1 \subseteq U_1$ by sampling every element of $U_1$ with probability $p_1=\min\left\{\frac{2^{j_1}}{\hat{\cR}},1\right\}$ independently of other elements. \\
	\For{($j_2=d\log n$ to $0$)}
	{
	find $B_2 \subseteq U_2$ by sampling every element of $U_2$ with probability $p_2=\min\left\{2^{j_2 - j_{1}} \cdot d\log n,1\right\}$ independently of other elements.
	
		{$\vdots \vdots$} 

		\For{($j_{d-1}=d\log n$ to $0$)}
		{
			find $B_{d-1} \subseteq U_{d-1}$ by sampling every element of $U_{d-1}$ with probability $p_{d-1}=\min\{2^{j_{d-1} - j_{d-2}} \cdot d\log n,1\}$ independently of other elements.\\
			
			Let ${\bf j}=(j_1,\ldots,j_{d-1}) \in [(d\log n)^*]^{d-1}$\\
			Let $p(i,{\bf{  j}}) =p_i$, where $1 \leq i \leq d-1$\\
			Let $B(i, {\bf j})=B_i$, where $1 \leq i \leq d-1$\\
			
			find $B(d, {\bf j})=B_{d} \subseteq U_{d}$ by sampling every element of $U_{d}$ with probability $p_{d}=\min\left\{{2^{-j_{d-1}}},1\right\}$ independently of other elements.\\

			\If{$\left( m_o(B_{1,{\bf j}},\ldots, B_{d,{\bf j}}) \neq 0  \right)$} 
			{
				{\sc Accept} {/*[Note that $\mbox{{\sc CID}}_2^o$ query is called in the above line.]*/}
			}
		}
	}
}

{\sc Reject}
\end{algorithm}
\begin{proof} Consider the set of ordered hyperedges $\cF_o(\cH_o)$ in $\cH_o$. Algorithm \verest taking parameters $\cH_o$, and $\hat{\cR}$ and described in Algorithm \ref{algo:verify}, loops over all possible ${\bf j}=(j_1,\ldots,j_{d-1})\in [(d \log n)^*]^{d-1}$~\footnote{Recall that $[n]^*$ denotes the set $\{0,\ldots,n\}$.}. For each ${\bf j}=(j_1,\ldots,j_{d-1})\in [(d \log n)^*]^{d-1}$, \verest($\cH_o,\hat{\cR}$) samples vertices in each $U_i$ with \emph{suitable} probability values $p(i,{\bf j})$, depending on ${\bf j}$, $\hat{R}$, $d$ and $\log n$, to generate the sets $B_{i,{\bf j}}$ for $1\leq i\leq d$. See Algorithm \ref{algo:verify} for the exact values of $p(i,{\bf j})$'s. \verest($\cH_o,\hat{\cR}$) reports {\sc Accept} if there exists one ${\bf j} \in [(d \log n)^*]^{d-1}$ such that $m_o\left(B_{1,{\bf j}},\ldots, B_{d,{\bf j}}\right)\neq 0$. Otherwise, {\sc Reject} is reported by \verest($\cH_o,\hat{\cR}$). 

For an ordered hyperedge $F_o \in \cF_o(\cH_o)=\cF_o(U_1,\ldots,U_d)$ and ${\bf j} \in \left[(d\log n)^*\right]^{d-1}$. Note that 
$$
    U_1=\ldots=U_d=U(\cH).
$$
Let $X^{{\bf j}}_{F_o}$ denote the indicator random variable such that $X^{\bf j}_{F_o}=1$ if and only if $F_o \in \cF_o(B_{1,{\bf j}}, \ldots, B_{d,{\bf j}})$. Let 
$$
    X_{\bf j}=\sum\limits_{F_o \in \cF_o(\cH_o)} X^{\bf j}_{F_o}. 
$$
Note that $m_o(B_{1,{\bf j}}, \ldots, B_{d,{\bf j}}) = X_{\bf j}$. We have, 
\begin{align*}
    \pr\left( X^{\bf j}_{F_o} =1\right) &= \prod\limits_{i=1}^d (p(i,{\bf j}))~\remove{See Algorithm~\ref{algo:verify} for the values of $p(i,{\bf j})$'s} \\
    &\leq \frac{2^{j_1}}{\hat{\cR}} \cdot \frac{2^{j_2}}{2^{j_1}} d\log n \times \cdots\times \frac{2^{j_{d-1}}}{2^{j_{d-2}}} d\log n\times \frac{1}{2^{j_{d-1}}}\\
    &= \frac{d^{d-2}\log ^{d-2} n}{\hat{\cR}}
\end{align*}
Then, 
$$
    \E\left[X_{\bf j}\right] \leq \frac{m_o(\cH_o)}{\hat{\cR}} d^{d-2}\log^{d-2} n,
$$    
and since $X_{\bf j} \geq 0$, we have
$$
    \pr \left( X _{\bf j}\neq 0\right)=\pr(X_{\bf j} \geq 1) \leq \E \left[X_{\bf j} \right] \leq \frac{m_o(\cH_o)}{\hat{\cR}} d^{d-2}\log ^{d-2} n.
$$
Now, using the fact that $\hat{\cR} \geq 20d^{2d-3} \cdot 4^{d} \cdot m_o(\cH_o) \log^{2d-3} n$, we have 
$$ \pr\left(X_{\bf j} \neq 0 \right) \leq \frac{1}{20d^{d-1} \cdot 4^{d} \cdot \log^{d-1}  n}. $$

Recall that \verest accepts if and only if there exists ${\bf j}$ such that $X_{\bf j} \neq 0$~\footnote{Note that ${\bf j}$ is a vector but $X_{\bf j}$ is a scalar.}. Using the union bound, we get
\begin{align*}
	\pr\left(\mbox{\verest$(\cH_o,\hat{\cR})$ accepts {the estimate $\hat{R}$}}\right)
	&\leq  \sum\limits_{{\bf j} \in [(d\log n)^*]^{d-1}} \pr(X_{\bf j} \neq 0)\\ 
	 &\leq  \frac{(d\log n +1)^{d-1}}{20 \cdot 4^d\cdot (d\log n)^{d-1}} \\
	 &\leq  \frac{1}{20\cdot 2^d}. 
\end{align*}	
\end{proof}

\begin{lem} \label{lem:coarse2} If $\hat{\cR} \leq \frac{m_o(\cH_o)} {4d \log n}$, 
$\pr( \mbox{\verest ($\cH_o,\hat{\cR}$) accepts {the estimate $\hat{R}$}}) \geq \frac{1}{2^d}.$ \end{lem}

\begin{proof}

 We will be done by showing the following. \verest accepts with probability at least $1/5$ when the loop variables $j_1,\ldots,j_{d-1}$ respectively attain values $q_1,\ldots,q_{d-1}$ such that 
 $$
    \size{U_1(q_1)} > \frac{m_o(\cH_o)}{2^{q_1+1}(d\log n +1)}
$$ 
 and 
 $$
 \size{ U_i((Q_i,A_i),q_i)} > \frac{2^{q_{i-1}}}{2^{q_i +1}(d\log n +1)}
 $$
 for all $i \in [d-1] \setminus \{1\}$. The existence of such $j_i$s is evident from Claim~\ref{clm:verify}. Let ${\bf q}=(q_1,\ldots,q_{d-1})$. Recall that $B_{i,{\bf q}} \subseteq U_i$ is the sample obtained when the loop variables $j_1,\ldots,j_{d-1}$ attain values $q_1,\ldots,q_{d-1}$, respectively. Let $\cE_i, i \in [d-1],$ be the events defined as follows.
\begin{itemize}
	\item $\cE_1~:~U_1(q_1) \cap B_{1,{\bf q}} \neq \emptyset$.
	\item $\cE_i~:~U_j((Q_{j-1},A_{j-1}),q_j) \cap B_{j,{\bf q}} \neq \emptyset$, where $2\leq i \leq d-1$.
\end{itemize} 
As noted earlier, Claim~\ref{clm:verify} says that for each $i \in [d-1]$, there exists a bucket in $U_i$ having a \emph{large} number of vertices contributing approximately the same number of hyperedges. The above events correspond to the nonempty intersection of vertices in heavy buckets corresponding to $U_i$ and the sampled vertices $B_{i,{\bf j}}$, where $i \in [d-1]$.
Observe that 
\begin{align*}
    \pr(\overline{\cE_1}) &\leq \left( 1- \frac{2^{q_1}}{\hat{\cR}}\right)^{\size{U_1(q_1)}}\\
    &\leq \exp {\left(-\frac{2^{q_1}}{\hat{\cR}}\size{U_1(q_1)}\right)}\\ &\leq  \exp{\left(-\frac{2^{q_1}}{\hat{\cR}}\cdot \frac{m_o(\cH_o)}{2^{q_1+1}(d \log n +1)}\right)}\\
    &\leq \exp{(-1)}.
\end{align*}  
The last inequality uses the fact that $\hat{\cR} \leq \frac{m_o(\cH_o)} {4d \log n}$, from the condition of the lemma. Assume that $\cE_1$ occurs and $a_1 \in U_1(q_1) \cap B_{1,{\bf q}} $. We will bound the probability that $U_2(Q_1,A_1),q_2) \cap B_{2,{\bf q}} = \emptyset$, that is $\overline{\cE_2}$. Note that, by Claim~\ref{clm:verify} (ii), 
$$ 
    \size{U_2(Q_1,A_1),q_2)} \geq \frac{2^{q_{1}}}{2^{q_2 +1}(d\log n +1)}.
$$
So,
$$ \pr\left(\overline{\cE_2}~|~\cE_1\right) \leq \left(1-\frac{2^{q_2}}{2^{q_1}} \log n \right)^{\size{U_2(Q_1,A_1),q_2)}} \leq \exp{(-1)}$$
Assume that $\cE_1,\ldots,\cE_{i-1}$ hold, where $3 \leq i \in [d-1]$. Let $a_1 \in U_1(q_1)$ and $a_{i-1} \in A_{i-1}((Q_{i-2},U_{i-2}),q_{i-1})$. We will bound the probability that $U_i((Q_{i-1},A_{i-1}),q_i) ~\cap~ B_{i,{\bf q}} = \emptyset$, that is $\overline{\cE_i}$. Note that 
$$
    \size{U_i((Q_{i-1},A_{i-1}),q_i)} \geq \frac{2^{q_{i-1}}}{2^{q_{i} +1}(d\log n +1)}.
$$
So, for $3 \leq i \in [d-1]$,
$$ \pr\left( \overline{\cE_i}~|~\cE_1,\ldots,\cE_{i-1}\right) \leq \left(1-\frac{2^{q_i}}{2^{q_{i-1}}} \log n \right)^{\size{U_i(Q_{i-1},A_{i-1}),q_i)}} \leq  \exp{(-1)}$$
Assume that $\cE_1,\ldots,\cE_{d-1}$ hold. Let $a_1 \in U_1(q_1)$ and $a_{i-1} \in A_{i-1}((Q_{i-2},A_{i-2}),q_{i-1})$ for all $i \in [d] \setminus \{1\}$. Let $S \subseteq U_d$ be the set of $d$-th vertex of the ordered hyperedges in $\cF_o(\cH_o)$ having $u_j$ as the $j$-th vertex for all $j \in [d-1]$. Note that $\size{S} \geq 2^{q_{d-1}}$. Let $\cE_d$ be the event that represents the fact $S \cap B_{d,{\bf q}} \neq \emptyset$. So,
$$ \pr(\overline{\cE_d}~|~\cE_1,\ldots,\cE_{d-1}) \leq \left(1- \frac{1}{2^{q_{d-1}}} \right)^{q_{d-1}} \leq \exp{(-1)}$$ 
Observe that \verest accepts if $m(B_{1,{\bf q}},\ldots,B_{d,{\bf q}}) \neq 0$. Also, 
$$m_o(B_{1,{\bf q}},\ldots,B_{d,{\bf q}}) \neq 0~\mbox{if}~\bigcap\limits_{i=1}^d \cE_i~ \mbox{occurs}.$$ Hence,
\begin{align*}
	\pr(\mbox{\verest$(\cH_o, \hat{\cR})$ accepts}) &\geq \pr\left( \bigcap\limits_{i=1}^d \cE_i \right)\\
	&= \pr(\cE_1)\prod\limits_{i=2}^d \pr\Bigg(\cE_i~\Big|~\bigcap\limits_{j=1}^{i-1}\cE_j\Bigg)\\
	&> \left(1-\frac{1}{e}\right)^{d}\\
	&> \frac{1}{2^d}.
\end{align*}
\end{proof}

Now, we will prove Lemma~\ref{lem:prob1} that will be based on Algorithm~\ref{algo:coarse}. 
\begin{algorithm}[!h]
\caption{$\cAc (\cH_o (U, \cF_o))$}
\label{algo:coarse}
\KwIn{$\mbox{{\sc CID}}_2^o$ query access to a $d$-uniform hypergraph $\cH_o(U, \cF_o)$.}
\KwOut{An estimate $\hat{m}_o$ for $m_o=m_o(\cH_o)$.}
\For{$(~\hat{\cR}= n^d,n^d/2,\ldots, 1)$}
{
	Repeat \verest $(\cH_o,\hat{\cR})$ for $\Gamma=d\cdot 4^d \cdot 2000 \log n $ times.
	If more than $\frac{\Gamma}{10 \cdot 2^d}$ \verest accepts, then output ${\hat{m}_o}=\frac{\hat{\cR}}{d^{d-2} \cdot   2^d\cdot (\log n)^{d-2}}$.
}
\end{algorithm}
\begin{proof}[Proof of Lemma~\ref{lem:prob1}] Note that an execution of {\cAc} for a particular $\hat{\cR}$ repeats \verest for $\Gamma =d \cdot 4^d \cdot 2000 \log n$ times and gives output $\hat{\cR}$ if more than $\frac{\Gamma}{10 \cdot 2^d}$ \verest accepts. For a particular $\hat{\cR}$, let $X_i$ be the indicator random variable such that $X_i=1$ if and only if the  $i$-th execution of \verest accepts. Also take $X=\sum_{i=1}^\Gamma X_i$. {\cAc} gives output $\hat{\cR}$ if $X > \frac{\Gamma}{10 \cdot 2^d}$.

Consider the execution of {\cAc} for a particular $\hat{\cR}$. If $\hat{\cR}  \geq 20d^{2d-3} 4^d \cdot m_o(\cH_o) \cdot$ $\log ^{2d-3} n$, then we first show that {\cAc} does not accept with high probability. Recall Lemma~\ref{lem:coarse1}. If $\hat{\cR} \geq 20d^{2d-3}4^d \cdot m_o(\cH_o)\log ^{2d-3} n$, $\pr(X_i =1) \leq \frac{1}{20 \cdot 2^d}$ and hence $\E[X] \leq \frac{\Gamma}{20\cdot 2^d}$. By using Chernoff-Hoeffding's inequality~(See  {Lemma~\ref{lem:cher_bound2}}~(i) in Section~\ref{sec:prelim}), 

$$ \pr \left(X > \frac{\Gamma}{10 \cdot 2^d} \right) =\pr\left( X > \frac{\Gamma}{20 \cdot 2^d} + \frac{\Gamma}{20 \cdot 2^d}\right) \leq \frac{1}{n^{10d}}$$ 

Using the union bound for all $\hat{\cR}$, the probability that {\cAc} outputs some $\hat{m}_o=\frac{\hat{\cR}}{d^{d-2}\cdot 2^d}$ such that $\hat{\cR} \geq 20d^{2d-3}4^d \cdot m_o(\cH_o)\log^{2d-3}n$, is at most $\frac{d \log n}{n^{10}}$. Now consider the instance when the for loop in the algorithm {\cAc} executes for a $\hat{\cR}$ such that $\hat{\cR} \leq \frac{m_o(\cH_o)}{ 4d \log  n}$. In this situation, $\pr(X_i=1) \geq \frac{1}{2^d}$. So, $\E[X] \geq \frac{\Gamma}{2^d}$. By using Chernoff-Hoeffding's inequality~(See  {Lemma~\ref{lem:cher_bound2}}~(ii) in Section~\ref{sec:prelim}), 

$$ \pr\left(X \leq \frac{\Gamma}{10 \cdot 2^d } \right) \leq \pr\left(X < \frac{\Gamma}{2^d} -\frac{4}{5} \cdot \frac{\Gamma}{ 2^{d}} \right) 	\leq \frac{1}{{n^{100d}}}$$  

By using the union bound for all $\hat{\cR}$, the probability that {\cAc} outputs some $\hat{m}_o=\frac{\hat{\cR}}{d^{d-2}\cdot 2^d}$ such that $\hat{\cR} \leq \frac{m_o(\cH_o)}{ 4d \log  n}$, is at most $\frac{d \log n}{n^{100d}}$. Observe that, the probability that {\cAc} outputs some $\hat{m}_o=\frac{\hat{\cR}}{d^{d-2}\cdot 2^d}$ such that $\hat{\cR}\geq 20 d^{2d-3}4^d m_o(\cH_o)\log ^{2d-3} n$ or $\hat{\cR} \leq \frac{m_o(\cH_o)}{4d \log n}$, is at most 
$$
    \frac{d\log n}{n^{10d}} +\frac{d\log n}{n^{100d}} \leq \frac{1}{n^{8d}}.
$$
Putting everything together, {\cAc} gives some $\hat{m}_o=\frac{\hat{\cR}}{d^{d-2} \cdot 2^d \cdot (\log n)^{d-2}}$ as the output with probability at least $1-\frac{1}{n^{8d}}$ satisfying 

$$ \frac{m_o(\cH_o)}{8d^{d-1}2^d \log^{d-1} n} \leq \hat{m}_o\leq 20d^{d-1}2^d \cdot m_o(\cH_o) \log^{d-1} n$$ 

From the pseudocode of \verest (Algorithm~\ref{algo:verify}), we call for \gtwopis queries 
only at line number {12}. In the worst case, \verest executes line number 12 for each ${\bf j} \in [(d\log n)^*]$. That is, the query complexity of \verest is 
$\Oh(\log ^{d-1} n)$. From the description of {\cAc}, {\cAc} calls \verest 
$\Oh_d(\log n)$ times for each choice of $\hat{R}$. Hence, {\cAc} makes $
\Oh_d(\log ^{d+1} n)$ $\mbox{{\sc CID}}_2^o$ queries.
\end{proof}


%% file: whyimp.tex
\section{Proof of Theorem~\ref{coro:imp}
}
\label{sec:whyimp}

Before getting into the reasons \emph{why Theorem~\ref{coro:imp} follows from Theorem~\ref{theo:main}}, let us first review the algorithms for \hest and \hsample by Dell \etal ~\cite{soda/DellLM20}.

\paragraph*{Overview of  Dell \etal~\cite{soda/DellLM20}}

\remove{Recall that we den $n$ and $m$ be the number of vertices and hyperedges in the input $d$-uniform hypergraph.
\begin{obs}\label{obs:basic}
Let us consider a $\Oh\left(\frac{1}{\eps^2}\log n \right)$ independent subhypergraphs each induced by $n/2$ random vertices, selected uniformly and independently. Denoting $X$ by the sum of the number of subhypergrahs present in $\Oh\left(\frac{1}{\eps^2}\log n \right)$ subhypergraphs, observe that $\frac{2^d}{t}X$ is an $\left(1 \pm \eps\right)$-approximation of $m$.
\end{obs}

If we repeat the procedure recursively $\Oh(\log n)$ times, then all the subhypergraphs will have a bounded number of vertices in terms of $d$, in which the number of hyperedges can be determined exactly by using $f(d)$ queries. But the number of induced subhypergraphs will be very \emph{large} $\Omega\left((\log n)^{\log n}\right)$ at the end.

Though  Observation~\ref{obs:basic}  is simple, but their algorithm is involved, and it involves
} 

Dell \etal's algorithm for \hsample make repeated calls to \hest. Their algorithm for \hest calls mainly three subroutines over $\Oh_d(\log n)$ iterations:  {\sc Coarse}, {\sc Halving}, and {\sc Trim}. {\sc Halving} and {\sc Trim} calls {\sc Coarse} repeatedly. So, {\sc Coarse} is the main building block for their algorithms for \hest and \hsample.

\paragraph*{{\sc Coarse} algorithm} 

It estimates the number of hyperedges in the hypergraph up to polylog factors by using polylog queries. The result is formally stated as follows, see \cite[Sec.~4]{soda/DellLM20}.

\begin{lem}[{\bf {\sc Coarse Algorithm} by Dell \etal~\cite{soda/DellLM20}}]\label{pro:coarse}
There exists an algorithm {\sc Coarse}, that has \gpis query access to a hypergraph $\cH(U,\cF)$, makes $\Oh_d\left(\log ^{2d+3} n \right)$ \gpis queries, and finds $\hat{m}$ satisfying 
$$
    \Omega_d\left(\frac{1}{\log^d n}\right) \leq \frac{\hat{m}}{m} \leq \Oh_d\left(\log^d n\right)
$$ 
with probability at least $1-1/\poly$.
\end{lem}
\begin{rem}\label{rem:coarse}
The objective of {\sc Coarse} algorithm by Dell \etal ~is essentially same as that our {\cAc} algorithm. Both of them can estimate the number of hyperedges in any induced subhypergrah. However, note that {\cAc} (as stated in Theorem~\ref{theo:main}) has better approximation guarantee and better query complexity than that of {\sc Coarse} algorithm of Dell \etal~(as stated in Lemma~\ref{pro:coarse}).
\end{rem}

The framework of Dell \etal~implies that the query complexity of \hest and \hsample can be expressed by the approximation guarantee and the query complexity of the {\sc Coarse} algorithm. This is formally stated as follows:

\begin{lem}[{\bf \hest and \hsample in terms of quality of {\sc Coarse} algorithm~\cite{soda/DellLM20}}]\label{pro:dell-est}
Let there exists an algorithm {\sc Coarse}, that has \gpis query access to a hypergraph $\cH(U,\cF)$, makes $q$ \gpis queries, and finds $\hat{m}$ satisfying $\frac{1}{b} \leq \frac{\hat{m}}{m} \leq b$ with probability at least $1-1/\poly$. Then 
\begin{description}
\item[(i)]  \hest can be solved by using 
$$
    \Oh_d\left(\log ^2 n\left(\log nb+ \frac{b^2 \log^2 n}{\eps^2}\right)q \right)
$$ 
\gpis queries.

\item[(ii)] \hsample can be solved by using 
$$
    \Oh_d\left(\log ^6 n\left(\log nb+ \frac{b^2 \log^2 n}{\eps^2}\right)q \right)
$$ 
\gpis queries.

\end{description}

\end{lem}

\paragraph*{Why Theorem~\ref{coro:imp} follows from Theorem~\ref{theo:main}?}

Observe that we get Proposition~\ref{pro:dell} (the result of Dell \etal) from Lemma~\ref{pro:coarse} by substituting $b=\Oh_d\left(\log ^d n \right)$ and $q=\Oh_d\left(\log ^{2d+3} n\right)$ in Lemma~\ref{pro:dell-est}. In 
Theorem~\ref{coro:imp}
we improve on the Proposition~\ref{pro:dell} by using our main result (Theorem~\ref{theo:main}), and substituting $b=\Oh_d\left(\log ^{d-1} n\right)$ and $q=\Oh_d \left( \log ^{d+2} n \right)$ in Lemma~\ref{pro:dell-est}. 

The main reason we get an improved query complexity for hyperedge estimation in Theorem~\ref{coro:imp} as compared to Dell \etal~(Proposition~\ref{pro:dell-est}) is our {\sc Rough Estimation} algorithm is an improvement over the {\sc Coarse} algorithm of Dell \etal~\cite{soda/DellLM20} in terms of approximation guarantee as well as query complexity.

\paragraph*{How  our {\cAc} improves over {\sc Coarse} of Dell \etal~\cite{soda/DellLM20}?}


{At a very high level, the frameworks of our {\cAc} algorithm and that of Dell \etal's {\sc Coarse} algorithm might look similar, but the main ideas involved are different.} Our {\cAc} (as stated in Lemma~\ref{lem:prob1}) directly deals with the hypergraph (though the ordered one) and makes use of $\mbox{{\sc CID}}_2^o$ queries. Note that each $\mbox{{\sc CID}}_2^o$ query can be simulated by using $\Oh_d(\log n)$ \gpis queries. However, {\sc Coarse} algorithm of Dell \etal~considers $\Oh_d(\log n)$ independent random $d$-partite hypergraphs by partitioning the vertex set into $d$ parts uniformly at random, works on the $d$-partite hypergraphs, and reports the median, of the $\Oh_d(\log n)$ outputs corresponding to random $d$-partite subhypergrahs, as the final output. So, there is $\Oh_d (\log n)$ blowup in both our {\cAc} algorithm and Dell \etal's {\sc Coarse} algorithm, though the reasons behind the blowups are different.

Our {\cAc} calls repeatedly ($\Oh_d(\log n)$ times) {\sc Verify Estimate} for each guess, where the total number of guesses is $\Oh_d(\log n)$.  In the {\sc Coarse} algorithm, Dell \etal~uses repeated calls $\left(\Oh_d\left(\log ^{d+1} n\right)\right)$ times to an analogous routine of our {\sc Verify Estimate}, which they name {\sc Verify Guess},  $\Oh_d(\log n)$ times. Their {\sc Verify Guess} has the following criteria for any guess $M$:
\begin{itemize}
    \item 
        If $M \geq \frac{d^d \log^{2d} n}{2^{3d-1}}m$, {\sc Verify Guess} accepts $M$ with probability at most $p$;
    
    \item   
        If $M \leq m$, {\sc Verify Guess} accepts $M$ with probability at least $2p$;
    
    \item 
        It makes $\Oh_d\left(\log ^d n\right)$ \gpis queries.
\end{itemize}
Recall that the number of \gtwopis queries made by each call to {\sc Verify Estimate} is $\Oh_d(\log ^{d-1} n)$, that is, $\Oh_d\left(\log ^d n\right)$ \gpis queries. So, in terms of the number of \gpis queries, both our {\cAc} and {\sc Coarse} of Dell \etal~have the same complexity. 

The probability $p$ in {\sc Verify Guess} of Dell \etal~\cite{soda/DellLM20} satisfies $p \approx_{d} \frac{1}{\log ^ d n}$, where $\approx_d$ is used suppress the terms involving $d$. So, {for each guess $M$ , their {\sc Coarse} algorithm has to call $\Oh_d\left(\frac{1}{p} \log n \right)=\Oh_d\left(\log^{d+1} n\right)$ times to distinguish decide whether it is the case $M \leq m$ or  $M \geq \frac{d^d \log^{2d} n}{2^{3d-1}}m$, with a probability at least $1-1/\poly$.} So, the total number of queries made by the {\sc Coarse} algorithm of Dell \etal~\cite{soda/DellLM20} is 
$$\Oh_d(\log n)\cdot \Oh_d(\log n) \cdot \Oh_d\left(\log ^{d+1} n\right) \cdot \Oh_d\left(\log ^d n\right) =\Oh_d\left(\log ^{2d+3} n\right).$$

The first $\Oh_d(\log n)$ term is due to the blow up incurred to convert original hypergraph to $d$-partite hypergraph, the second $\Oh_d(\log n)$ term is due to the number of guesses for $m$, the third $\Oh_d\left(\log ^{d+1} n\right)$ term is the number of times {\sc Coarse} calls {\sc Verify Guess}, and the last term $\Oh_d\left(\log ^d n\right)$ is the number of \gpis queries made by each call to {\sc Verify Guess}.

As it can be observed from Lemmas~\ref{lem:coarse1} and~\ref{lem:coarse2}, $p$ in our case ({\sc Verify Estimate}) is $\Omega_d(1)$. So, it is enough for {\cAc} to call {\sc Verify Estimate} only $\Oh_d(\log n)$ times. Therefore, the number of \gpis queries made by our {\cAc} is 
$$
    \Oh_d(\log n) \cdot \Oh_d(\log n) \cdot \Oh_d(\log ^{d-1} n) \cdot \Oh_d(\log n)=\Oh_d(\log ^{d+2} n).
$$
In the above expression, the first $\Oh_d(\log n)$ term is due to the number of guesses for $m$, the second $\Oh_d\left(\log n\right)$ term is the number of times {\cAc} calls {\sc Verify Estimate}, the third $\Oh\left(\log ^{d-1} n\right)$ term is the number of \gtwopis queries made by each call to {\sc Verify Estimate}, and the last $\Oh_d(\log n)$ term is the number of \gpis queries needed to simulate a \gtwopis query with probability at least $1-1/\poly$.

We do the improvement in approximation guarantee as well as query complexity in {\sc Rough Estimation} algorithm (as stated in Theorem~\ref{theo:main}), as compared to {\sc Coarse} algorithm of Dell \etal~\cite{soda/DellLM20} (as stated in Lemma~\ref{pro:coarse}), by a careful analysis of the intersection pattern of the hypergraphs and setting the sampling probability parameters in {\sc Verify Estimate} (Algorithm~\ref{algo:verify}) algorithm in a nontrivial  way, which is evident from the description of Algorithm~\ref{algo:verify} and its analysis.



\newpage

%% file: appendix.tex
\input{probability}

\input{oracle-def}

%% file: probability.tex
\section{Some probability results} \label{sec:prelim}

\begin{lem}[Chernoff-Hoeffding bound~\cite{DubhashiP09}]
\label{lem:cher_bound1}
Let $X_1, \ldots, X_n$ be independent random variables such that $X_i \in [0,1]$. For $X=\sum\limits_{i=1}^n X_i$ and $\mu=\E[X]$, the followings hold for any $0\leq \delta \leq 1$.
$$ \pr(\size{X-\mu} \geq \delta\mu) \leq 2\exp{\left(-\mu \delta^2/3\right)}$$

\end{lem}
\begin{lem}[Chernoff-Hoeffding bound~\cite{DubhashiP09}]
\label{lem:cher_bound2}
Let $X_1, \ldots, X_n$ be independent random variables such that $X_i \in [0,1]$. For $X=\sum\limits_{i=1}^n X_i$ and $\mu_l \leq \E[X] \leq \mu_h$, the followings hold for any $\delta >0$.
\begin{description}
\item[(i)] $\pr \left( X > \mu_h + \delta \right) \leq \exp{\left(-2\delta^2/n\right)}$.
\item[(ii)] $\pr \left( X < \mu_l - \delta \right) \leq \exp{\left(-2\delta^2 / n\right)}$.
\end{description}
\end{lem}

%% file: oracle-def.tex
\section{Oracle definitions} \label{sec:oracle-def}

\begin{defi}
[Independent set query ({\sc IS})~\cite{BeameHRRS18}] Given a subset $A$ of the vertex set $V$ of a graph $G(V,E)$, {\sc IS} query answers whether $A$ is an independent set.
\end{defi}

\begin{defi}
[Bipartite independent set oracle (\bis)~\cite{BeameHRRS18}] Given two disjoint subsets $A,B$ of the vertex set $V$ of a graph $G(V,E)$, \bis query reports  whether there exists an edge having endpoints in both $A$ and $B$.
\end{defi}

\begin{defi}
[Tripartite independent set oracle (\tis)~\cite{BhattaISAAC}] Given three disjoint subsets $A,B,C$ of the vertex set $V$ of a graph $G(V,E)$, the {\sc TIS} oracle reports whether there exists a triangle having endpoints in $A,B$ and $C$.
\end{defi}

\begin{defi}
[Generalized $d$-partite independent set oracle (\gpis)~\cite{BishnuGKM018}] Given $d$ pairwise disjoint subsets of vertices $\dsubset \subseteq U(\cH)$ of a hypergraph $\cH$ as input, \gpis{} query answers whether $m(\dsubset) \neq 0$, where $m(\dsubset)$ denotes the number of hyperedges in $\cH$ having exactly one vertex in each $A_i$, $\forall i \in \{1,2,\ldots,d\}$. \end{defi}

\begin{defi}
[\gonepis oracle] Given $s$ pairwise disjoint subsets of vertices $A_1,\ldots,A_s \subseteq U(\cH)$ of a hypergraph $\cH$ and $a_1,\ldots,a_s \in [d]$ such that $\sum_{i=1}^{s}a_i =d$, \gonepis query on input $A_1^{[a_1]},A_2^{[a_2]},\cdots,A_s^{[a_s]}$ answers whether $m(\ssubset) \neq 0$.
\end{defi}

\begin{defi}
[\gtwopis oracle] Given any $d$ subsets of vertices $\dsubset \subseteq U(\cH)$ of a hypergraph $\cH$, \gtwopis query on input $A_1,\ldots,A_d$ answers whether $m(\dsubset) \neq 0$.
\end{defi}

\begin{defi}
[$\mbox{{\sc CID}}_2^o$ oracle] Given any $d$ subsets of vertices $\dsubset \subseteq U(\cH_o)$ of an ordered hypergraph $\cH_o$, $\mbox{{\sc CID}}_2^o$ query on input $A_1,\ldots,A_d$ answers  {\sc Yes}  if and only if $m_o(\dsubset) \neq 0$.
\end{defi}